\documentclass[11pt,reqno]{amsart}
\usepackage{amsmath,amssymb,amsthm,enumerate,natbib,color,ifthen,graphicx,overpic,hyperref}
\usepackage{xargs}
\usepackage[textwidth=4cm, textsize=footnotesize]{todonotes}
\usepackage[latin1]{inputenc}

\usepackage{aliascnt,bbm}
\def\rset{\mathbb R}
\def\zset{\mathbb Z}

\def\eqsp{\;}

 \newcommand{\seq}[1]{\left\langle#1\right\rangle}

\newcommand{\eqdef}{\ensuremath{\stackrel{\mathrm{def}}{=}}}

\def\Xset{\mathsf{X}} 
\def\F{\mathcal{F}} 

\newcommandx\sequence[3][2=t,3=\zset]
{\ifthenelse{\equal{#3}{}}{\ensuremath{\{ #1_{#2}\}}}{\ensuremath{\{ #1_{#2}, \eqsp #2 \in #3 \}}}}

\def\PP{\mathbb{P}} 
\newcommand{\CPP}[3][]
{\ifthenelse{\equal{#1}{}}{{\mathbb P}\left(\left. #2 \, \right| #3 \right)}{{\mathbb P}_{#1}\left(\left. #2 \, \right | #3 \right)}}
\def\PE{\mathbb{E}} 
\newcommand{\CPE}[3][]
{\ifthenelse{\equal{#1}{}}{{\mathbb E}\left[\left. #2 \, \right| #3 \right]}{{\mathbb E}_{#1}\left[\left. #2 \, \right | #3 \right]}}

\def\I{\textsf{I}}
\def\Id{\mathbb{I}}

\usepackage{bm}

\theoremstyle{plain}
\newtheorem{theorem}{Theorem}
\newtheorem{assumption}{H\hspace{-3pt}}

\newaliascnt{proposition}{theorem}
\newtheorem{proposition}[proposition]{Proposition}
\aliascntresetthe{proposition}

\newaliascnt{lemma}{theorem}

\aliascntresetthe{lemma}
\newaliascnt{corollary}{theorem}
\newtheorem{corollary}[corollary]{Corollary}
\aliascntresetthe{corollary}

\theoremstyle{definition}
\newaliascnt{definition}{theorem}

\aliascntresetthe{definition}

\newtheorem{algorithm}{Algorithm}

\newaliascnt{remark}{theorem}
\newtheorem{remark}[remark]{Remark}
\aliascntresetthe{remark}

\newaliascnt{example}{theorem}

\aliascntresetthe{example}

\def\rmd{\mathrm{d}}

\def\1{\mathbbm{1}}

\begin{document}

\title[On the sparse Bayesian Learning of linear models]{On the sparse Bayesian Learning of linear models}
\author{Chia Chye Yee}  \thanks{C. C. Yee: University of Michigan, 1085 South University, Ann Arbor,
  48109, MI, United States. {\em E-mail address:} chye@umich.edu}
\author{Yves F. Atchad\'e}  \thanks{ Y. F. Atchad\'e: University of Michigan, 1085 South University, Ann Arbor,
  48109, MI, United States. {\em E-mail address:} yvesa@umich.edu}



\maketitle

\begin{center} (Dec. 2014) \end{center}

\begin{abstract}
This work is a re-examination of the sparse Bayesian learning (SBL) of linear regression models of \cite{tipping01}  in a high-dimensional setting. We propose a hard-thresholded version of the SBL estimator that achieves, for orthogonal design matrices, the non-asymptotic estimation error rate of $\sigma\sqrt{s\log p}/\sqrt{n}$, where $n$ is the sample size, $p$ the number of regressors, $\sigma$ is the regression model standard deviation, and $s$ the number of non-zero regression coefficients. We also establish that with high-probability the estimator identifies the non-zero regression coefficients.   In our simulations we found that sparse Bayesian learning regression performs better than lasso (\cite{tib96}) when the signal to be recovered is strong.
\end{abstract}

\bigskip

\setcounter{secnumdepth}{3}

\section{Introduction}
High-dimensional variable selection has become an important topic in modern statistics. The least absolute shrinkage and selection
operator (lasso) of \cite{tib96} is probably the most widely used method for this problem and has span an extensive literature (see e.g. the monograph \cite{buhlmannetdegeer11}). Despite its success, the method has many shortcomings. For instance choosing the right amount of regularization remains a difficult and computer-intensive issue for many models. In parallel to the frequentist approach, Bayesian variable selection for high-dimensional problems has also generated a large literature (see for instance \cite{oharaetsillanpaa09} and the reference therein). But most Bayesian variable selection methods often lead to intractable posterior distributions that require a heavy use of Markov Chain Monte Carlo simulation. Between these two well-established frameworks lies an empirical Bayes alternative known as sparse Bayesian learning (SBL, \cite{tipping01}), which has received much less attention in the statistical literature. 

This paper is a re-examination of the SBL for linear regression in a high-dimensional setting. An interesting question is whether the SBL procedure recovers the sparsity structure of underlying signals. This problem was considered by \cite{wipfetrao04} which establishes that in the noiseless setting the SBL indeed recovers the sparsity structure of the regression coefficients. However the method behaves differently in  a noisy setting. For  orthogonal design matrices, we show that the SBL indeed produces a sparse solution of the regression coefficients, but does not in general recover the sparsity structure of the regression coefficients. To remedy this limitation we propose a hard-thresholded version of the SBL estimator. We show that with high probability the thresholded estimator achieves the same estimation error of $O(\sigma\sqrt{s\log(p)/n})$ as lasso, where $n$ is the sample size, $\sigma$ is the regression model standard deviation, $p$ the number of regressors and $s$ the number of non-zero regression coefficients. Furthermore we show that with high probability this thresholded estimator recovers the sparsity structure of the regression coefficients provided that the signal is not too weak.

Finally we did a simulation study comparing SBL and lasso. We find that the performance of SBL depends on the strength of the signal (defined here as the minimum of the absolute value of the non-zero coefficients).  With a weak signal SBL performs poorly compared to lasso, but outperforms lasso when the signal is strong.

The paper is organized as follows. We introduce the SBL method at the beginning of Section \ref{sec:main}. We study the computation and the sparsity structure of the SBL estimator in Sections \ref{sec:existence}-\ref{sec:computing}. The hard-thresholded estimator is defined and studied in Section \ref{sec:stat:theory}. The simulation study is reported in Section \ref{sec:sim}, and all the technical proofs are grouped in Section \ref{sec:proofs}. We end the paper with some open problems in Section \ref{sec:conclusion}.

\section{Sparse Bayesian learning of linear regression models}\label{sec:main}
Suppose that we observe a vector $y\in\rset^n$ that is a realization of a random variable $Y$ such that
\begin{equation}\label{model1}
Y=X\beta_\star + \epsilon,\end{equation}
for a known and non-random design matrix $X\in\rset^{n\times p}$, a vector $\beta_\star\in\rset^p$, and a random error term $\epsilon\in\rset^n$ such that \begin{equation}\label{assump_model}
\PE(\epsilon)=0, \;\; \mbox{ and }\;\; \PE(\epsilon\epsilon')=\sigma_\star^2 \Id_n,\end{equation}
for $\sigma^2_\star>0$,  where $\Id_n$ is the $n$-dimensional identity matrix. Our objective is to estimate $\beta_\star$ and $\sigma_\star^2$. Although (\ref{model1}-\ref{assump_model}) does not make any specific distributional assumption on $Y$,  we will consider the following possibly misspecified model: $Y\sim \textbf{N}(X\beta,\sigma^2I_n)$, with parameter $(\beta,\sigma^2)\in\rset^p\times (0,\infty)$, where $\textbf{N}(\mu,\Sigma)$ denotes the Gaussian distribution with mean $\mu$ and covariance matrix $\Sigma$. The parameter $\sigma^2$ is taken as fixed, and we assign to $\beta$ a prior distribution of the form
\begin{equation}\label{prior_dist}
\pi_\gamma(\rmd \beta)\eqdef \prod_{j=1}^p p_{\gamma_j}(\rmd \beta_j).\end{equation}
for a (hyper)-parameter $\gamma=(\gamma_1,\ldots,\gamma_p)\in\Theta\eqdef [0,\infty)^p$, where for $a>0$, $p_a$ denotes the distribution of $\textbf{N}(0,a)$, the Gaussian distribution on $\rset$ with mean $0$ and variance $a$,  and $p_0(\rmd u)\eqdef\delta_0(\rmd u)$ denotes the Dirac measure at $0$. 
 The posterior distribution of $\beta$ given $Y=y$ and given the hyper-parameter $(\gamma,\sigma^2)$ is therefore
\begin{equation}\label{postdist}
\pi_n(\rmd \beta\vert y,\sigma^2,\gamma)\propto \left(\frac{1}{2\pi\sigma^2}\right)^{n/2}\exp\left(-\frac{1}{2\sigma^2}\|y-X\beta\|^2\right)\pi_\gamma(\rmd \beta).\end{equation}

Sampling from the posterior distribution $\pi_n(\cdot\vert y,\sigma^2,\gamma)$ is straightforward. Indeed, for $\gamma=(\gamma_1,\ldots,\gamma_p)\in\Theta$, denote $\textsf{I}_\gamma\eqdef \{1\leq j\leq p:\; \gamma_j\neq 0\}$ the sparsity structure defined by $\gamma$. Notice that for $j\notin\I_\gamma$ (that is $\gamma_j=0$), $\pi_\gamma$ puts probability mass $1$ on the event $\{\beta_j=0\}$, and so does the posterior distribution $\pi_n(\cdot\vert y,\sigma^2,\gamma)$. Hence $\pi_n(\cdot\vert y,\sigma^2,\gamma)$ is the distribution of the random variable $(B_1,\ldots,B_p)$ obtained by simulating $\{B_j,j\in \textsf{I}_\gamma\}$ from $\textbf{N}(\mu_\gamma,\sigma^2V_\gamma)$, and by setting the remaining components to $0$, where
\begin{equation}\label{moment:post}
\mu_\gamma=V_\gamma X_\gamma'y,\;\;V_\gamma=\left(X_\gamma'X_\gamma+\sigma^2\bar\Gamma^{-1}_\gamma\right)^{-1},\end{equation}
where $X_\gamma$ is the matrix obtained from $X$ by removing the columns $j$ for which $\gamma_j=0$, and $\bar \Gamma_\gamma$ is the diagonal matrix with diagonal elements given by $\{\gamma_j,\,j\in\textsf{I}_\gamma\}$. With this Gaussian linear model, and prior (\ref{prior_dist}), it is easy to check that the marginal distribution of $y$ is $\textbf{N}(0,C_\gamma)$, where
\[C_\gamma\eqdef \sigma^2\Id_n +\sum_{j\in\I_\gamma}\gamma_jx_jx_j',\]
and $x_j$ is the $j$-th column of $X$. Therefore, up to a normalizing constant that we ignore, the log-likelihood of $(\sigma^2,\gamma)$ is given by
\[\ell(\sigma^2,\gamma)\eqdef -\frac{1}{2}\log\det(C_\gamma)-\frac{1}{2}\textsf{Tr}\left(C_\gamma^{-1}yy'\right).\]
The sparse Bayesian learning (SBL) estimator of $\beta_\star$ as proposed by \cite{tipping01,faulettipping02} is the empirical Bayes estimator of $\beta$ given by
\begin{equation}\label{estimator1}
\hat \beta_n=\int\beta\pi_n(\rmd \beta\vert y,\hat\sigma_n^2,\hat\gamma_n),\end{equation}
where
\begin{equation}\label{opt:probl}
(\hat\sigma_n^2,\hat\gamma_n)=\textsf{Argmax}_{(\sigma^2,\gamma)\in\rset_+\times\Theta}\;\ell(\sigma^2,\gamma).\end{equation}
Notice that $\hat\beta_n$ is straightforward to compute once $\hat\sigma_n^2$ and $\hat\gamma_n$ are available. Indeed given $\hat\sigma_n^2$ and $\hat\gamma_n$, $\hat\beta_{n,j}=0$ for all $j$ such  that $\hat\gamma_{n,j}=0$, and for the other components $j\in \I_{\hat\gamma_n}$, we have from (\ref{moment:post}) that
\[(\hat\beta_{n,j})_{j\in \I_{\hat\gamma_n}}=\left(X_{\hat\gamma_n}'X_{\hat\gamma_n}+\hat\sigma_n^2\bar\Gamma^{-1}_{\hat\gamma_n}\right)^{-1}X_{\hat\gamma_n}'y.\]

\begin{remark}
The presentation of the SBL given above is slightly different from the original presentation of \cite{tipping01,faulettipping02}. The key difference here is that in the prior distribution $\pi_\gamma$ we allow the components of $\gamma$ to take the value zero. This is needed for the estimator $\hat\gamma_n$ to be well-defined, and for the well-posedness of the question of whether the procedure produces sparse solutions.
\vspace{-0.2cm}
\begin{flushright}
$\square$
\end{flushright}
\end{remark}

Computationally, the optimization problem (\ref{opt:probl}) is not a ``nice" problem because the objective function $\ell(\sigma^2,\gamma)$ is non-concave and typically attains its maximum at the boundary of the domain $\Theta$ (that is some of the components of its solution(s) are exactly zeros). We return to the issue of solving (\ref{opt:probl}) in Section \ref{sec:computing}. But statistically (\ref{opt:probl}) is interesting as it yields a sparse solution $\hat\gamma_n$ as we shall see.

\subsection{Existence of $\hat\gamma_n$}\label{sec:existence}
Since the log-likelihood function $\ell$ is not concave in general, it is not immediately clear that the optimization problem (\ref{opt:probl}) has a solution. It is even less clear whether the solution is sparse.  Focusing on the case where $\sigma^2$ is assumed known, we show that a solution always exists.
\begin{proposition}\label{prop1}
Fix $y\in\rset^n$, $X\in\rset^{n\times p}$, and $\sigma^2=\sigma^2_\star$. Then the maximization problem $\textsf{Argmax}_{\gamma\in\Theta}\ell(\gamma,\sigma^2)$ has at least one solution $\hat\gamma=(\hat\gamma_{1},\ldots,\hat\gamma_{p})$ which has the following property: 
\begin{equation}\label{expr:hatgamma}
\hat\gamma_{j}=\left\{\begin{array}{ll}\frac{\left(x_j'C_j^{-1}y\right)^2-x_j'C_j^{-1}x_j}{\left(x_j'C_j^{-1}x_j\right)^2} & \;\;\mbox{ if }\;\; \left(x_j'C_j^{-1}y\right)^2> x_j'C_j^{-1}x_j\\ 0 &\mbox{ otherwise}\;,\end{array}\right.\end{equation}
where $C_j$ is given by
\[C_j\eqdef \sigma^2\Id_n +\sum_{k\in\I_{\hat\gamma}\setminus\{j\}}\hat\gamma_{k}x_kx_k'.\]
\end{proposition}
\begin{proof}
See Section \ref{proof:prop1}.
\end{proof}

It is important to notice that there is no randomness involved in the above result: $y$ and $X$ are given and fixed. In particular we do not assume (\ref{model1}) nor (\ref{assump_model}). It is clear that this result does not give the expression of the maximizer since the right-hand side of (\ref{expr:hatgamma}) also depends on $\hat\gamma$. Rather it gives coherence relationships between components of the solution. But more importantly the proposition shows that the optimization problem  (\ref{opt:probl}) leads to sparse solutions $\hat\gamma$. One can interpret the term $x_j'C_{j}^{-1}y$ as a measure of correlation between the $y$ and the $j$-th column $x_j$ of $X$. Hence the result shows that if the correlation between $x_j$ and $y$ is sufficiently weak then $\hat\gamma_{n,j}$ (and hence $\hat\beta_{n,j}$) is set exactly equal to zero. Of course  Proposition \ref{prop1} is useful only to the extent that the inequality $\left(x_j'C_j^{-1}y\right)^2\leq x_j'C_j^{-1}x_j$ is satisfied with high probability when $\beta_{\star,j}=0$. We investigate this below. Unfortunately we will see that in general $\hat\gamma$ does not recover exactly  the sparsity structure of $\beta_\star$, even in the most favorable setting. We make the following distributional assumption.
\begin{assumption}
\label{A1} The data generating model (\ref{model1}-\ref{assump_model}) holds and $\epsilon\sim \textbf{N}(0,\sigma_\star^2I_n)$, for some $\sigma_\star^2>0$.
\end{assumption}
\medskip

We shall also focus our analysis on the idealized case where the matrix $X$ has orthogonal columns.

\begin{assumption}
\label{A2} The design matrix $X\in\rset^{n\times p}$ is such that $\seq{x_k,x_j}=0$ whenever $j\neq k$.
\end{assumption}
\medskip

\begin{proposition}\label{prop:prob}
Suppose that H\ref{A1}-\ref{A2} hold, and $\sigma^2=\sigma_\star^2$. Then for any $j\in\{1,\ldots,p\}$ such that $\beta_{\star,j}=0$, 
\[\PP\left[\hat\gamma_{n,j}=0\right] =  \PP\left[Z^2\leq 1\right]\approx 0.68,\]
where $Z\sim \textbf{N}(0,1)$.
\end{proposition}
\begin{proof}
See Section \ref{proof:prop:prob}.
\end{proof}

The result above shows that even in the idealized setting of H\ref{A2}, and under the Gaussian linear model assumption, the SBL procedure will set $\hat\gamma_j$ to $0$ (for $j\notin \I$) only about $70\%$ of the time, regardless of the sample size. We do not know whether this result continue to hold for more general design matrices. The behavior of the solution of (\ref{opt:probl}) for a general design matrix is technically more challenging.

Another important limitation of the SBL procedure is the computation of $\hat\gamma_n$ and $\hat\sigma^2_n$. Typically iterative methods (such as the EM algorithm, see Section \ref{sec:computing}) are used. The EM algorithm does not promote sparsity, and converges to the solution only at the limit. Therefore, in finite time, the solutions generated by the EM algorithm are typically not sparse at all.

These two shortcomings limit the usefulness of the basic SBL procedure as an interesting method for sparse signal recovery. However, we observe that when $\beta_{\star,j}=0$, and the condition $\left(x_j'C_{j,\hat\gamma_n}^{-1}Y\right)^2\leq x_j'C_{j,\hat\gamma_n}^{-1}x_j$ fails, assuming again the most favorable setting of H\ref{A2}, $\hat\gamma_j$ is given 
\[\hat\gamma_{n,j}=\frac{\sigma^2(Z_j^2-1)}{\seq{x_j,x_j}},\]
where $Z_j\sim \textbf{N}(0,1)$. Hence $\hat\gamma_j$ has mean zero and variance of order $O(\|x_j\|^{-4})\approx O(n^{-2})$. We  conclude that when SBL fails to set to zero a component $j$ such that $\beta_{\star,j}=0$, the computed SBL solution $\hat\gamma_{j}$ is typically very small. This suggests that a thresholded version of $\hat\gamma_n$ should be able to set  these terms to zero. We pursue this approach in Section \ref{sec:stat:theory}. 


\subsection{A thresholded version and its statistical properties}\label{sec:stat:theory}
We saw in Section \ref{sec:existence} that although sparse, $\hat\gamma_n$  does not recover in general the sparsity structure of $\beta_\star$. To improve on this we propose a modified, hard-thresholded version of $\hat\gamma_n$ denoted $\tilde\gamma_n$ and defined as follows. For $1\leq j\leq p$,
\begin{equation}\label{thres_est}
\tilde\gamma_{n,j}\eqdef\left\{\begin{array}{cc}\hat\gamma_{n,j} & \mbox{ if } \hat\gamma_{n,j}>\frac{\hat\sigma_n^2 z_\star}{\|x_j\|^2}\\ 0  & \mbox{ otherwise}\end{array}\right.,\end{equation}
for a thresholding parameter $z_\star$ that we set to $z_\star=c(1+|\hat \rho|)\log p$, for a constant $c$, and where $\hat\rho$ is an estimate of the largest correlation among the columns of $X$. The corresponding modified  estimator of $\beta_\star$ is 
\[\tilde \beta_n\eqdef\int\beta\pi_n(\rmd \beta\vert y,\hat\sigma_n^2,\tilde\gamma_n).\]

\begin{theorem}\label{thm1}
Assume H\ref{A1}-\ref{A2}, and suppose that $\sigma^2_\star$ is known, $\log s\geq 1$, and $z_\star=c_0\log p$, for some constant $c_0> 2$, where $s=|\I_{\gamma_\star}|$. Then
\begin{equation}
\|\tilde \beta_n-\beta_\star\|_2^2 \leq M\frac{\sigma^2 s\log(p)}{n},
\end{equation}
with probability at least $1-\frac{1}{p^{(c_0 s)/8}}-\frac{1}{\exp(s)}$, where $M=\frac{4(2+c_0)}{c}$, and $c=\min_{1\leq i\leq p}\|x_j\|^2/n$.
\end{theorem}
\begin{proof}See Section \ref{proof:thm1}. \end{proof}

We deduce the following corollary. For $u\in\rset^p$, $\textsf{sign}(u)=(s_1,\ldots,s_p)$ where for each $i$, $s_i=0$ if $u_i=0$, $s_i=1$ is $u_i>0$, and $s_i=-1$ if $u_i<0$.

\begin{corollary}\label{coro1}
In addition to the assumptions of Theorem \ref{thm1}, suppose that 
\begin{equation}\label{cond:signal}
\min_{\{j:\; |\beta_{\star,j}|>0\}} |\beta_{\star,j}|> \sqrt{\frac{M\sigma^2s\log p}{n}}.\end{equation}
Then with probability at least $1-\frac{1}{p^{\frac{c_0s}{8}}} -\frac{1}{\exp(s)}-\frac{1}{p^{\frac{c_0}{2}-1}}$, $\textsf{sign}(\tilde\beta_n)=\textsf{sign}(\beta_\star)$.
\end{corollary}
\begin{proof}See Section \ref{proof:coro1}. \end{proof}

\subsection{Computing $\hat\sigma_n^2$ and $\hat\gamma_n$} \label{sec:computing}
Here we address the issue of solving (\ref{opt:probl}).  Because the function $\ell(\sigma^2,\gamma)$ is not concave, and typically attains its maximum at the boundary of the domain $\Theta$, the optimization  (\ref{opt:probl}) is not a smooth problem. The strategy originally developed by \cite{tipping01} focuses instead on the smooth problem obtained by maximizing  $\ell$ over the open domain $\rset_+^{p+1}$, where $\rset_+\eqdef (0,\infty)$. That is, find
\begin{equation}\label{probl2}
\displaystyle\textsf{Argmax}_{(\sigma^2,\gamma)\in\rset_+^{p+1}}\;\ell(\sigma^2,\gamma).\end{equation}
Of course,  this latter problem  has no solution whenever  the solution of (\ref{opt:probl}) occurs at the boundary of $\Theta$. Nevertheless, we will see that an EM algorithm that attempts to solve (\ref{probl2})  produces sequences that converge to the solution of (\ref{opt:probl}).

Since the likelihood function $\exp(\ell)$ of $(\sigma^2,\gamma)$ is obtained by integrating out $\beta$, we can treat $\beta$ as a missing variable and use the EM algorithm as proposed by \cite{tipping01}. For $\gamma\in\rset_+^p$, the so-called complete log-likelihood takes the form
\[\ell_\textsf{com}(\beta,\sigma^2,\gamma\vert y)=-\frac{n}{2}\log\sigma^2 -\frac{1}{2\sigma^2}\|y-X\beta\|^2-\frac{1}{2}\sum_{j=1}^p\left(\log\gamma_j + \frac{\beta_j^2}{\gamma_j}\right).\]
Given a working solution $(\{\sigma^2\}^{(k)},\gamma^{(k)})\in \rset_+^{p+1}=(0,\infty)^{p+1}$,  set
\[Q(\sigma^2,\gamma\vert \{\sigma^2\}^{(k)},\gamma^{(k)})\eqdef \int \ell_\textsf{com}(\beta,\sigma^2,\gamma\vert y)\pi_n(d\beta\vert y, \{\sigma^2\}^{(k)},\gamma^{(k)}),\]
the so-called $Q$-function. We will use the upper-script $(k)$ to index sequences generated by the EM algorithm. Set $V^{(k)}=\left(X'X+\{\sigma^2\}^{(k)}\{\Gamma^{(k)}\}^{-1}\right)^{-1}$, and $\mu^{(k)}=V^{(k)}X'y$, where $\Gamma^{(k)}=\textsf{diag}(\gamma^{(k)}_{1},\ldots,\gamma^{(k)}_{p})$. Maximizing $Q(\cdot\vert  \{\sigma^2\}^{(k)},\gamma^{(k)})$ is easy and gives 
\begin{eqnarray*}
\{\sigma^2\}^{(k+1)}&=&\frac{1}{n}\int \|y-X\beta\|^2\pi_n(\beta\vert \{\sigma^2\}^{(k)},\gamma^{(k)},y)\rmd\beta\\
&=&n^{-1}\left(\|y-X\mu^{(k)}\|^2+\{\sigma^2\}^{(k)}\textsf{Tr}(V^{(k)}X'X)\right),\end{eqnarray*}
and for $j=1,\ldots,p$,
\[\gamma_{j}^{(k+1)}=\{\mu_{j}^{(k)}\}^2+\{\sigma^2\}^{(k)}V_{j,j}^{(k)}.\]
This leads to the following algorithm for solving (\ref{probl2})
\begin{algorithm}[EM algorithm]\label{EMalgo}
Given $(\{\sigma^2\}^{(k)},\gamma^{(k)})\in \rset_+^{p+1}=(0,\infty)^{p+1}$, we compose the matrix $\Gamma^{(k)}=\textsf{diag}(\gamma_{1}^{(k)},\ldots,\gamma_{p}^{(k)})$.
\begin{enumerate} 
\item  Compute $V^{(k)}=\left(X'X+\{\sigma^2\}^{(k)}\{\Gamma^{(k)}\}^{-1}\right)^{-1}$, and $\mu^{(k)}=V^{(k)}X'y$.
\item Set  
\[\gamma_{j}^{(k+1)}=\{\mu_{j}^{(k)}\}^2+\{\sigma^2\}^{(k)}V_{j,j}^{(k)},\;\;j=1,\ldots,p,\] 
\[\{\sigma^2\}^{(k+1)}=\frac{1}{n}\left(\|y-X\mu^{(k)}\|^2+\{\sigma^2\}^{(k)}\textsf{Tr}(V^{(k)}X'X)\right).\]
\end{enumerate}
\end{algorithm}

Although this EM algorithm is designed to solve the maximization problem (\ref{probl2}) we will see that it typically converges to the solution of (\ref{opt:probl}). To simplify the analysis we assume again that H\ref{A2} holds and that $\sigma^2$ is fixed. Hence we focus only on the recursion in $\gamma$:
\[\gamma_{j}^{(k+1)}=\{\mu_{j}^{(k)}\}^2+\sigma^2V_{j,j}^{(k)},\;\;j=1,\ldots,p.\]
With the assumption that  the design matrix is orthogonal, we can work out explicitly the terms $V^{(k)}=\left(X'X+\sigma^2\{\Gamma^{(k)}\}^{-1}\right)^{-1}$ and $\mu^{(k)}=V^{(k)}X'y$, which leads to
\begin{equation}\label{recEM}
\gamma_{j}^{(k+1)}=\frac{\seq{x_j,y}^2}{\left(\|x_j\|^2+\frac{\sigma^2}{\gamma_{j}^{(k)}}\right)^2} + \frac{\sigma^2}{\|x_j\|^2+\frac{\sigma^2}{\gamma_{j^{(k)}}}},\;\;j=1,\ldots,p.\end{equation}

\begin{proposition}\label{prop2}
Fix $y\in\rset^n$, and $X\in\rset^{n\times p}$ such that H\ref{A2} holds. Fix $\sigma^2>0$. Let $\{\gamma^{(k)},\; k\geq 0\}$ denote the sequence produced by the recursion (\ref{recEM}) for some initial $\gamma^{(0)}$ with positive components. Then for all $j\in\{1,\ldots,p\}$,
\[\lim_{k\to \infty}\gamma_{j}^{(k)}=\hat\gamma_{n,j}=\left\{\begin{array}{ll}\frac{\seq{y,x_j}^2-\sigma^2\|x_j\|^2 }{\|x_j\|^2}& \mbox{ if } \; \seq{y,x_j}^2>\sigma^2\|x_j\|^2 \\ 0 & \mbox{ otherwise }\end{array}\right..\]
\end{proposition}
\begin{proof}
See Section \ref{proof:prop2}
\end{proof}

\begin{remark}
In the non-orthogonal design setting, our simulation results suggest that the conclusion of Proposition \ref{prop1} continues to hold, although we do not have any rigorous proof.
\end{remark}

\subsection{A simulation study}\label{sec:sim}
\subsubsection{Synthetic Data Sets}
We investigate by simulation the behavior of the SBL procedure and its thresholding version, and how they compare with lasso. For all the simulations, $n=100$ and $p=500$. We generate the design matrix $X$ by simulating each row independently from the Gaussian distribution $\textbf{N}(0,\Sigma)$ where $\Sigma_{ii}=1$ and $\Sigma_{ij}=\rho$ for $i\neq j$. We consider two values of $\rho$: $\rho=0$ for which $X$ is close to satisfy H\ref{A2}, and $\rho=0.9$ which produces a design matrix $X$ with strongly correlated variables. We simulate the dependent variable $Y$ from the $\textbf{N}(X\beta_\star,\sigma^2_\star \mathbb{I}_n)$, with $\sigma_\star=1$. We consider four (4) different scenarios of sparsity, with $s=3,15,25$, and $s=50$ where $s$ is the number of non-zero elements of $\beta_\star$. The magnitude of the non-zero elements also play an important role in the recovery. We generate all the non-zeros components of $\beta_\star$ from the uniform distribution $\textbf{U}(a,a+1)$, for $a$ ranging from $0$ to $9$. 

For each value of $\rho$, each sparsity level, and each signal strength $a$, we repeat each estimator  $30$ times, and  we compute  the relative error rate ($\|\hat\beta-\beta_\star\|/\|\beta_\star\|$), the sensitivity and the specificity, averaged over these $30$ replications.  The sensitivity (\textsf{SEN}) and the specificity (\textsf{SPE}) of a given estimator $\hat \beta$ are defined as 
\[
\textsf{SEN}(\hat\beta)=\frac{\sum_{j=1}^p \textbf{1}_{\{\hat{\beta}_j \neq 0\}}\textbf{1}_{\{\beta_{\star,j} \neq 0\}}}{\sum_{j=1}^p\textbf{1}_{\{\beta_{\star,j} \neq 0\}}},\;\;\mbox{ and }\;\;\textsf{SPE}(\hat\beta)=\frac{\sum_{j=1}^p \textbf{1}_{\{\hat{\beta}_j \neq 0\}}\textbf{1}_{\{\beta_{\star,j} \neq 0\}}}{\sum_{j=1}^p \textbf{1}_{\{\hat{\beta}_j \neq 0\}}}.\]

These measures are valid for any estimator $\hat\beta$, and we compute them for the thresholded version of SBL, the non-thresholded version of SBL, as well as for the lasso estimator. For the thresholded SBL, we use $z_\star=c(1+|\hat\rho|)\log p$, where $c$ is determined by minimizing the BIC: $\frac{\|y - X\hat{\beta} \|}{2\hat{\sigma}^2} + s \log (n) $.

We compute the lasso estimator using  the function $cv.glmnet$ of the package GLMNet (\cite{glmnet}) where we select the penalty term $\lambda$ by a $10$-fold cross-validation procedure. In the cross-validation, the regulation parameter selected minimizes the prediction error.

The simulation results are presented on Figure 1-8. As one can see from these figures,  the main conclusion is that SBL is more sensitive than lasso to the strength of the signal (defined here as as the parameter $a$).  With a weak signal it performs poorly, but outperforms lasso when the signal is strong enough. Another interesting finding is that, overall, lasso performs  poorly in selecting the non-zeros components (variable selection). This is consistent with recent results (\cite{meinhausenetyu09}) which shows that variable selection consistency of lasso requires the irrepresentable condition, which actually is a very strong condition that often does not hold in practice. For instance, the irrepresentable condition fails for all the design matrices of this simulation study, except for the design matrix behind Figure 2.



\subsubsection{A simulated real data example}
In this example, we consider a micro-array data concerning genes involved in the production of riboflavin. The data is made publicly available at \begin{verbatim}
http://www.annualreviews.org/doi/suppl
                           /10.1146/annurev-statistics-022513-115545\end{verbatim}
 and contains $n = 71$ samples and $p = 4088$ covariates corresponding to $4088$ genes. Each of the sample contains a real valued response consisting of the logarithm of the riboflavin production rate and $4088$ real valued covariates consisting of the logarithm of the genes' expression levels. 

Given the very high dimensionality of this dataset, the lack of any true value of the parameter, and given also the fact that micro-array data are well-known to be very noisy, direct comparison of different regression methods on such dataset cannot be very insightful. For a more meaningful comparison, we use the riboflavin design matrix $X \in \mathbb{R}^{71 \times 4088}$ to generate simulated levels of riboflavin production rate using the sparse regression model $Y = X \beta + \epsilon$ where $\epsilon \sim N(0,\sigma^2\mathbb{I}_{71})$, with $\sigma^2=1$. The magnitude of the non-zero components of $\beta$ are uniformly simulated $\beta_j \sim U(a, a+1)$ with $a = \{0,1, \dots, 9\}$. We set the number of non-zeros elements in the vector $\beta$ to $5$. Figure 9 shows the results of the simulation evaluated using the aforementioned metrics. Under such extreme high-dimensional conditions, both methods perform poorly. SBL has found all the relevant variables but has also selected many non-relevant variables. Lasso has produced more sparse solutions, but has missed some important variables. The results remain essentially the same even when we set $\sigma^2$ (the variance of the noise $\epsilon$) to $0.1$.



\vspace{0.2cm}
One final word on computing times. We compute the SBL estimate using Algorithm \ref{EMalgo}, and we use the package GLMNet to compute lasso. We implemented Algorithm \ref{EMalgo} in \textsf{R}. The core of the GLMNet package is written in Fortran and the result is very fast. The comparison of the computing times is largely in favor of GLMNet. Comparing computing times is always tricky as it depends to a large extent on the programming language and skills. But beyond the implementation differences, it seems clear that lasso has a computational advantage over SBL in that it leads to ``easier" (convex) optimization problems, compared to SBL.

\section{Conclusion}\label{sec:conclusion}
We have shown that when the design matrix is orthogonal, the SBL estimator is uniquely defined, sparse (however does not recover the true sparsity structure of the signal), and can be computed using the EM algorithm. We have also proposed a hard-thresholded version of SBL, and shown that the hard-thresholded estimator recovers the true sparsity structure of the model, and achieves the same estimation error bound as lasso (with high probability). Furthermore our simulation results show that the method compares very well with lasso, and outperforms lasso when the regression coefficients are not too small.

One important and pressing  issue is the extension of these results to non-orthogonal design matrices. In particular we wish to understand the type of design matrix $X$ for which these results continue to hold. This SBL theory and its comparison with the recently developed lasso theory (see for instance \cite{meinhausenetyu09,bickeletal09}) could potentially give new insight into high-dimensional regression analysis. The generalized singular value decomposition (see e.g. \cite{golubetvanloan13}) of $X_\gamma$ and $\Gamma_\gamma$ seems to be a promising approach to tackle this problem. The challenge in this approach appears to be the development of an appropriate differentiability theory for the components of the GSVD decomposition as a function of $\gamma$.

The SBL method can be extended in several directions. It can be easily extended to deal with generalized linear models, and graphical models. But in these extensions, the computation of the estimator might require some new algorithms. Another possible extension of the method would be to replace the Gaussian distribution in the prior $\pi_\gamma$ by some other distribution. Some of these extensions of the methodology are already being explored. For instance \cite{balaetmadigan09} replaced the Gaussian distribution by the double-exponential distribution and shows by simulation that the resulting estimator (called demi-lasso) compares very well with lasso and the standard SBL.

\begin{center}
{\normalsize \scalebox{0.7}{\includegraphics{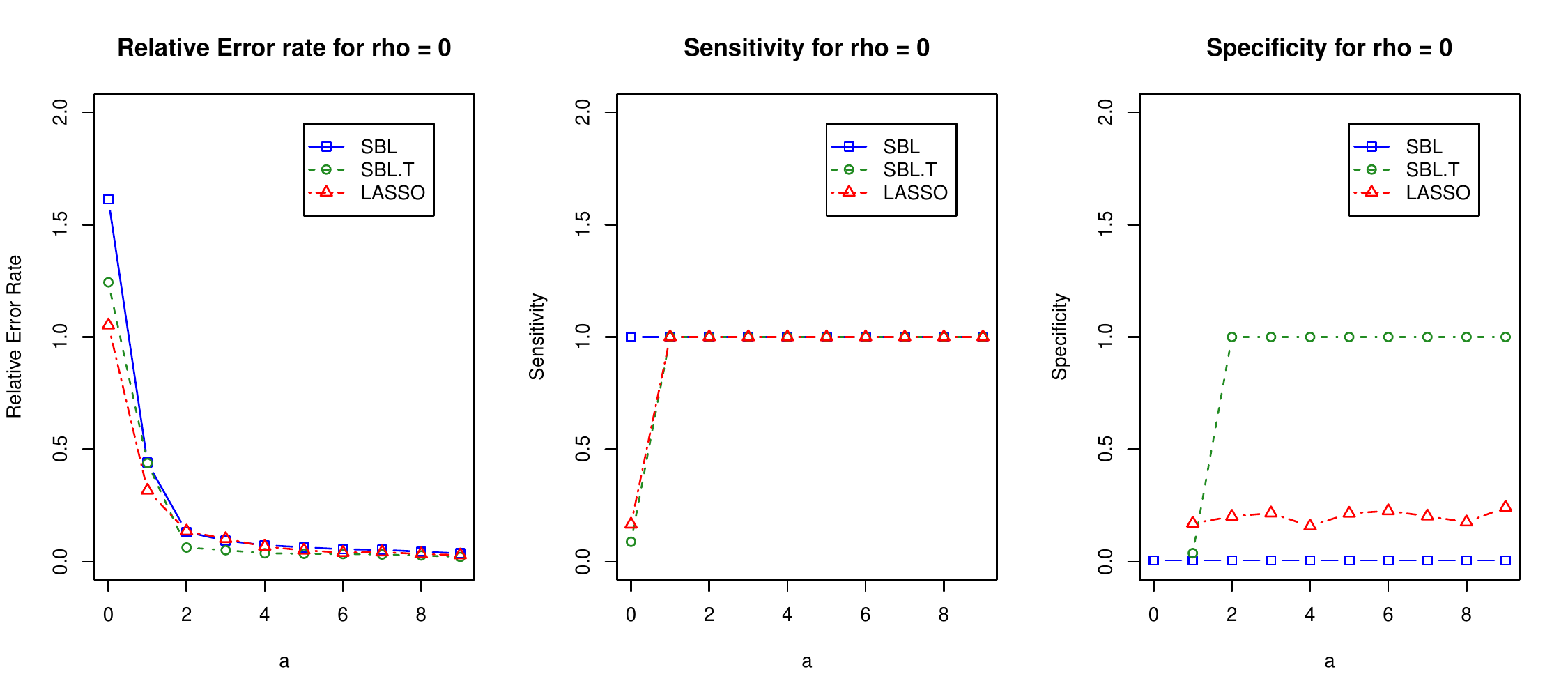}}\\[0pt]
\noindent\underline{Figure 1}: Sensitivity, specificity and relative error for SBL and lasso as function of $a$. $s=3$.}
\end{center}

\begin{center}
{\normalsize \scalebox{0.7}{\includegraphics{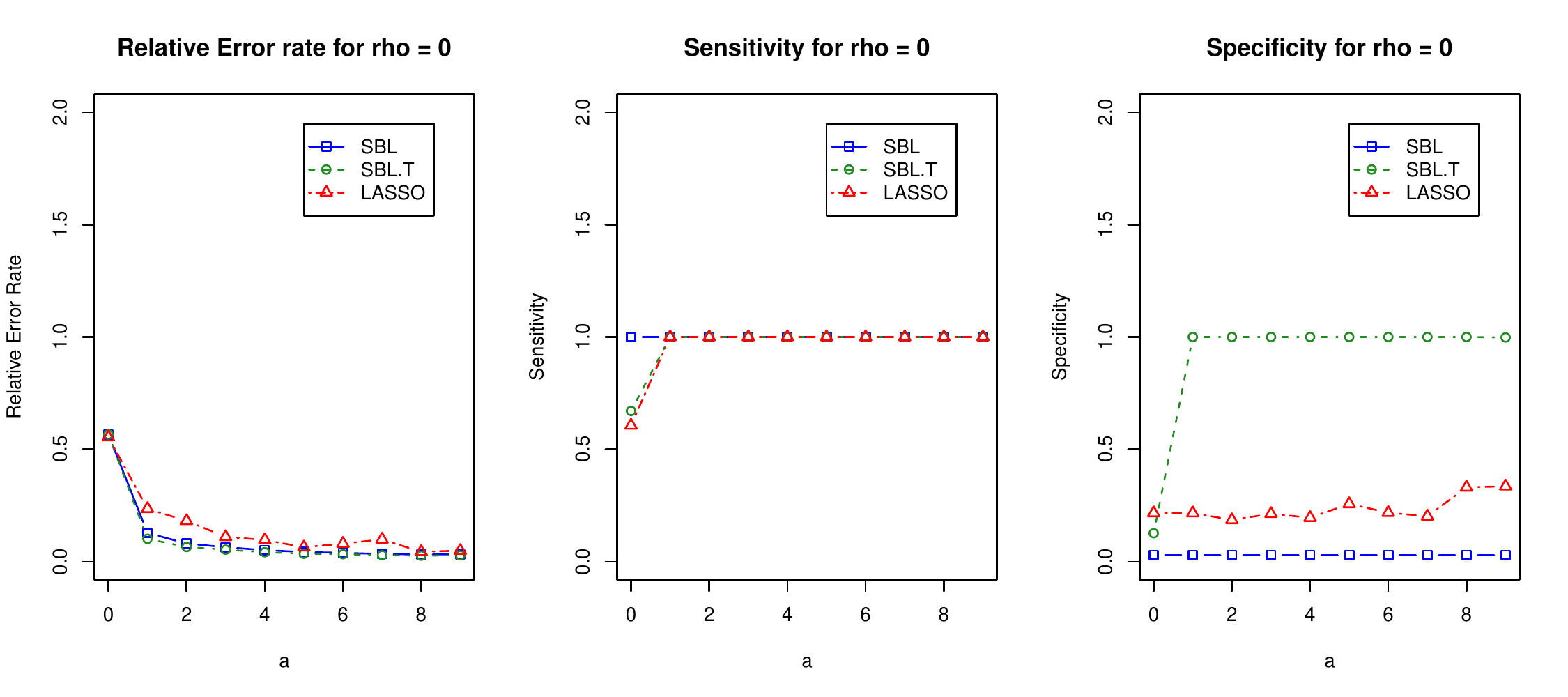}}\\[0pt]
\noindent\underline{Figure 2}: Sensitivity, specificity and relative error for SBL and lasso as function of $a$. $s=15$.}
\end{center}

\begin{center}
{\normalsize \scalebox{0.7}{\includegraphics{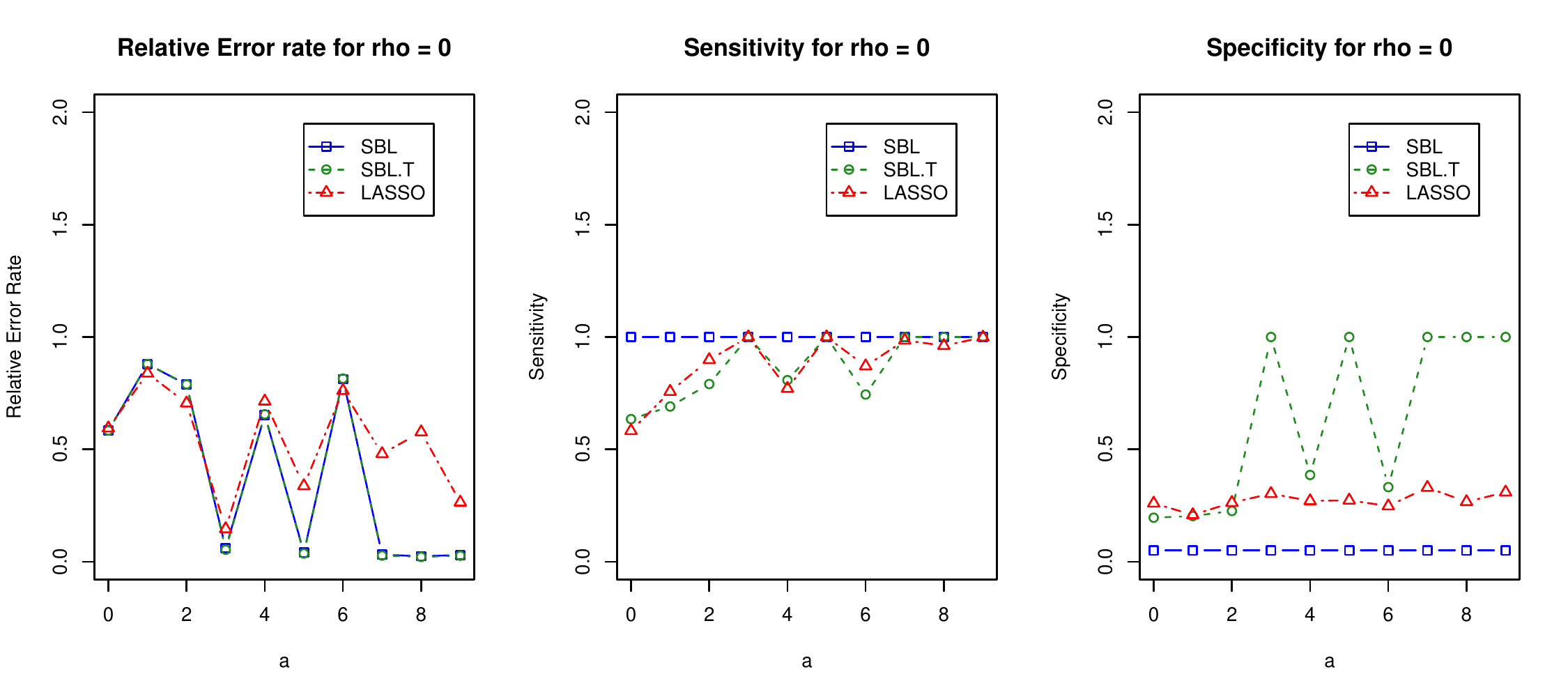}}\\[0pt]
\noindent\underline{Figure 3}: Sensitivity, specificity and relative error for SBL and lasso as function of $a$. $s=25$.}
\end{center}

\begin{center}
{\normalsize \scalebox{0.7}{\includegraphics{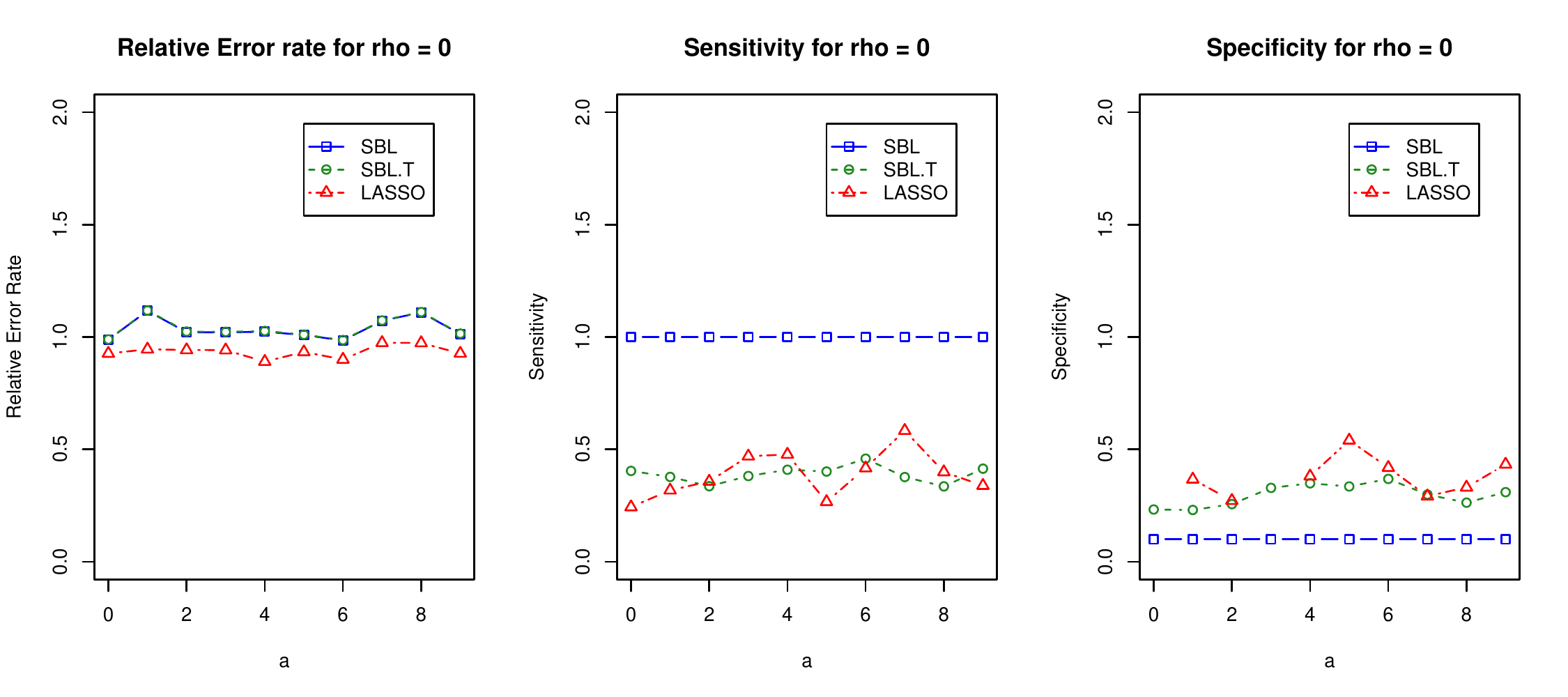}}\\[0pt]
\noindent\underline{Figure 4}: Sensitivity, specificity and relative error for SBL and lasso as function of $a$. $s=50$.}
\end{center}





\begin{center}
{\normalsize \scalebox{0.7}{\includegraphics{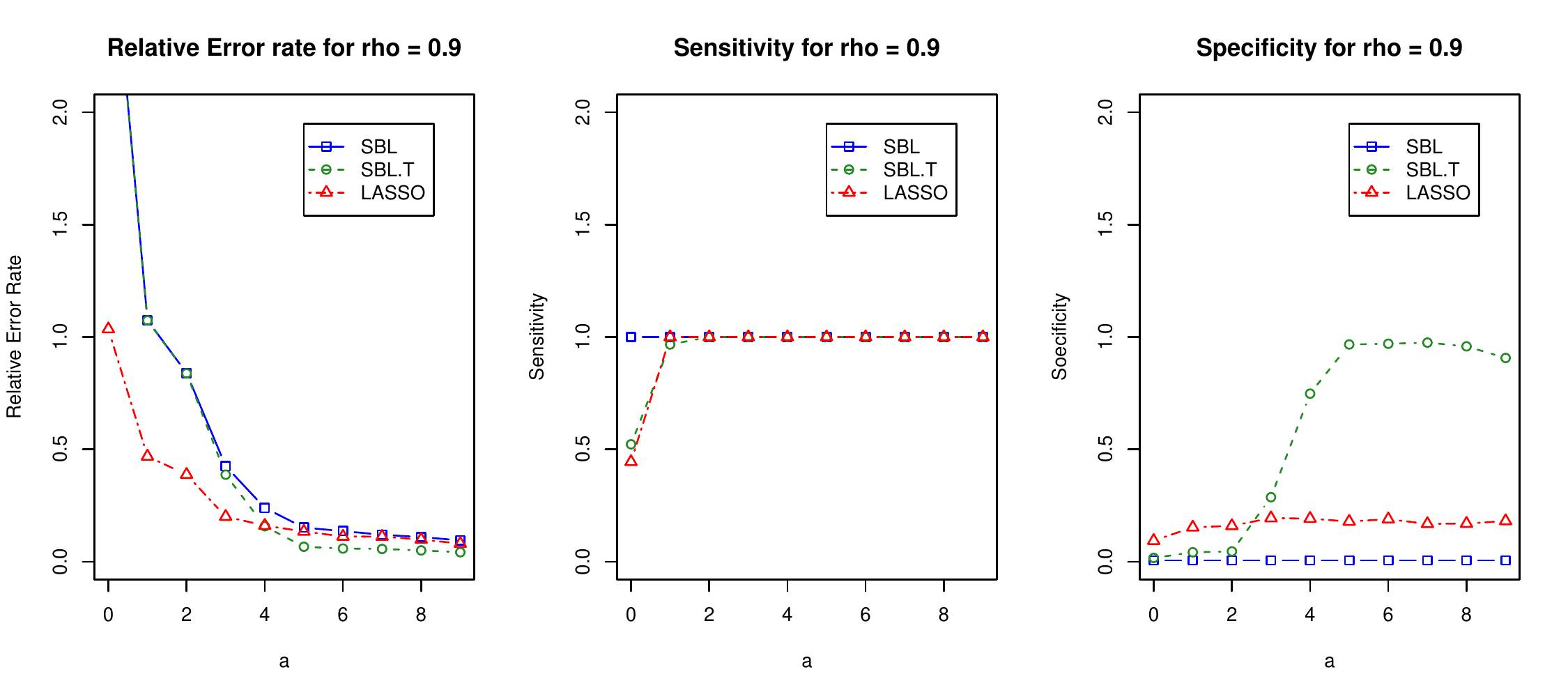}}\\[0pt]
\noindent\underline{Figure 5}: Sensitivity, specificity and relative error for SBL and lasso as function of $a$. $s=3$.}
\end{center}

\begin{center}
{\normalsize \scalebox{0.7}{\includegraphics{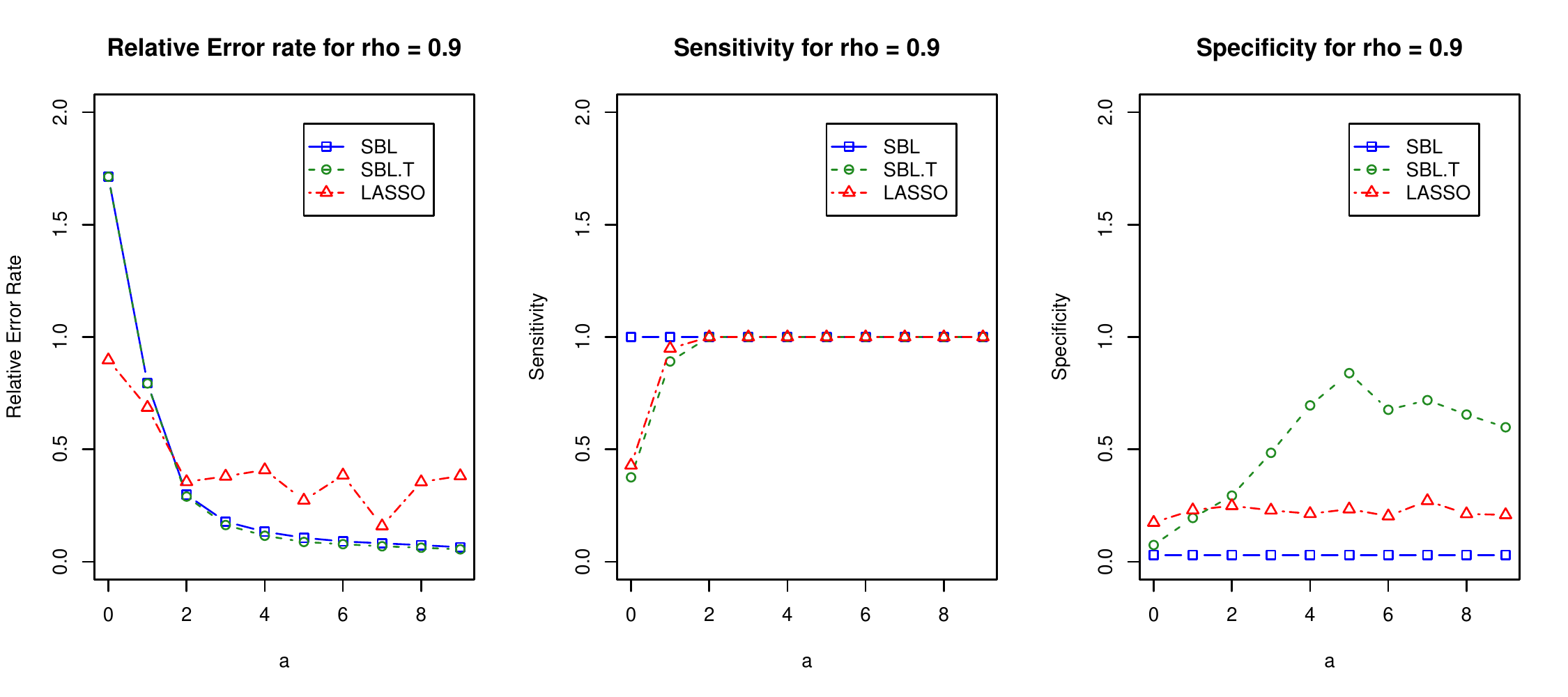}}\\[0pt]
\noindent\underline{Figure 6}: Sensitivity, specificity and relative error for SBL and lasso as function of $a$. $s=15$.}
\end{center}

\begin{center}
{\normalsize \scalebox{0.7}{\includegraphics{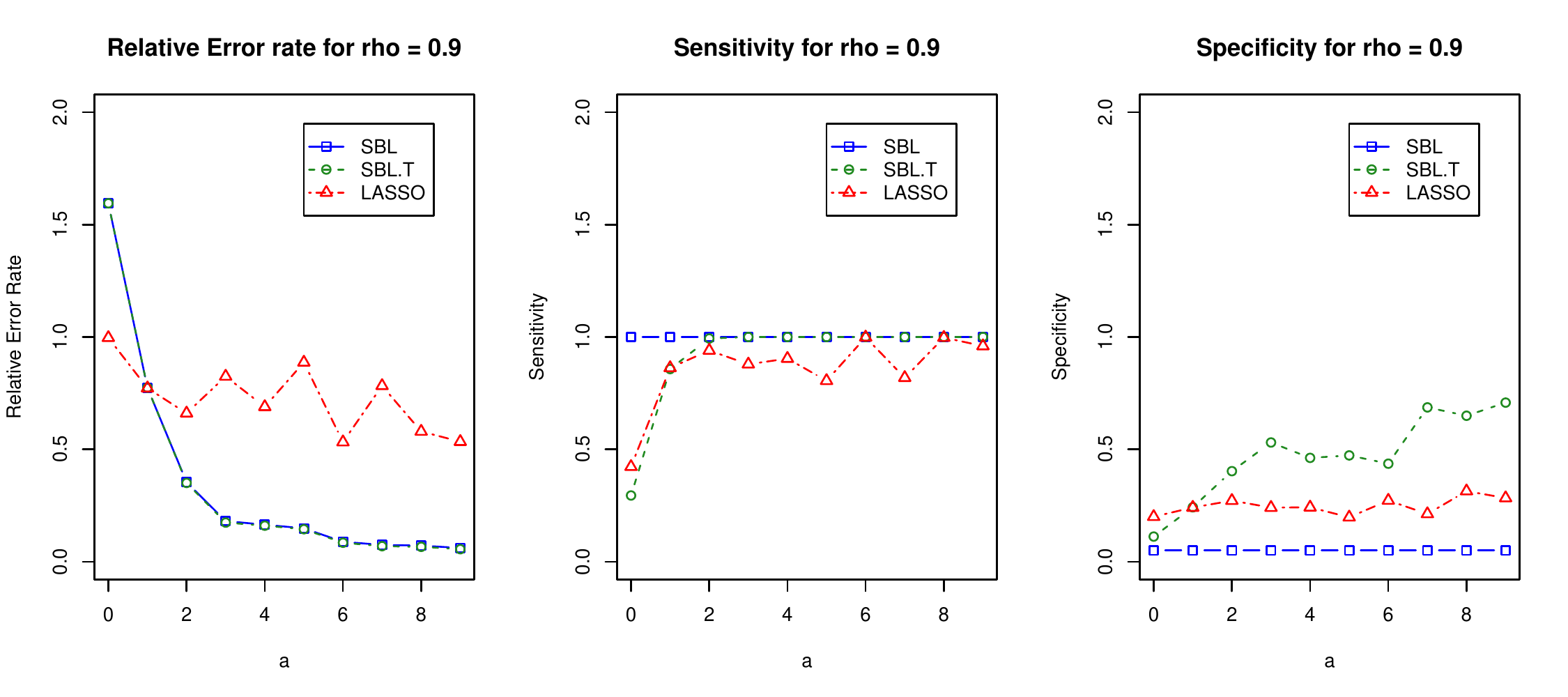}}\\[0pt]
\noindent\underline{Figure 7}: Sensitivity, specificity and relative error for SBL and lasso as function of $a$. $s=25$.}
\end{center}

\begin{center}
{\normalsize \scalebox{0.7}{\includegraphics{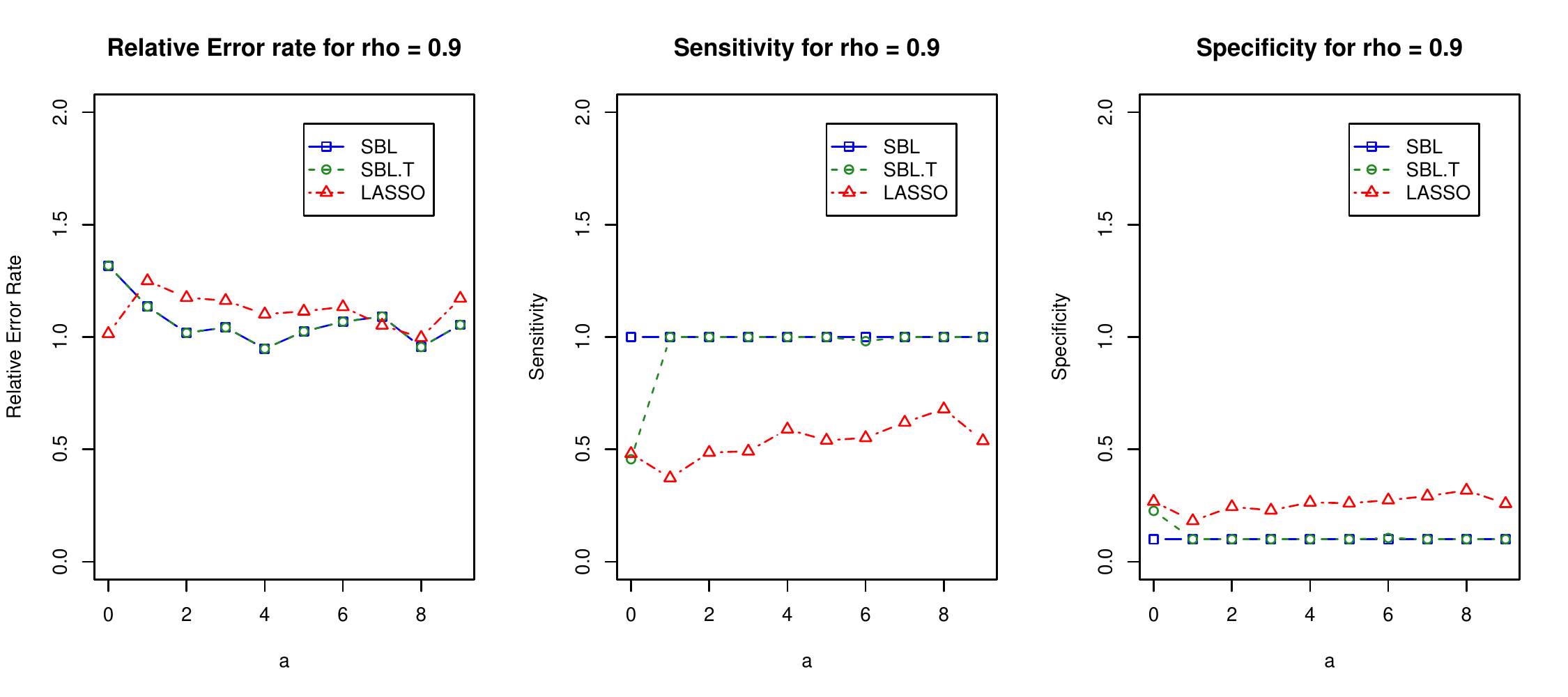}}\\[0pt]
\noindent\underline{Figure 8}: Sensitivity, specificity and relative error for SBL and lasso as function of $a$. $s=50$.}
\end{center}

\begin{center}
{\normalsize \scalebox{0.7}{\includegraphics{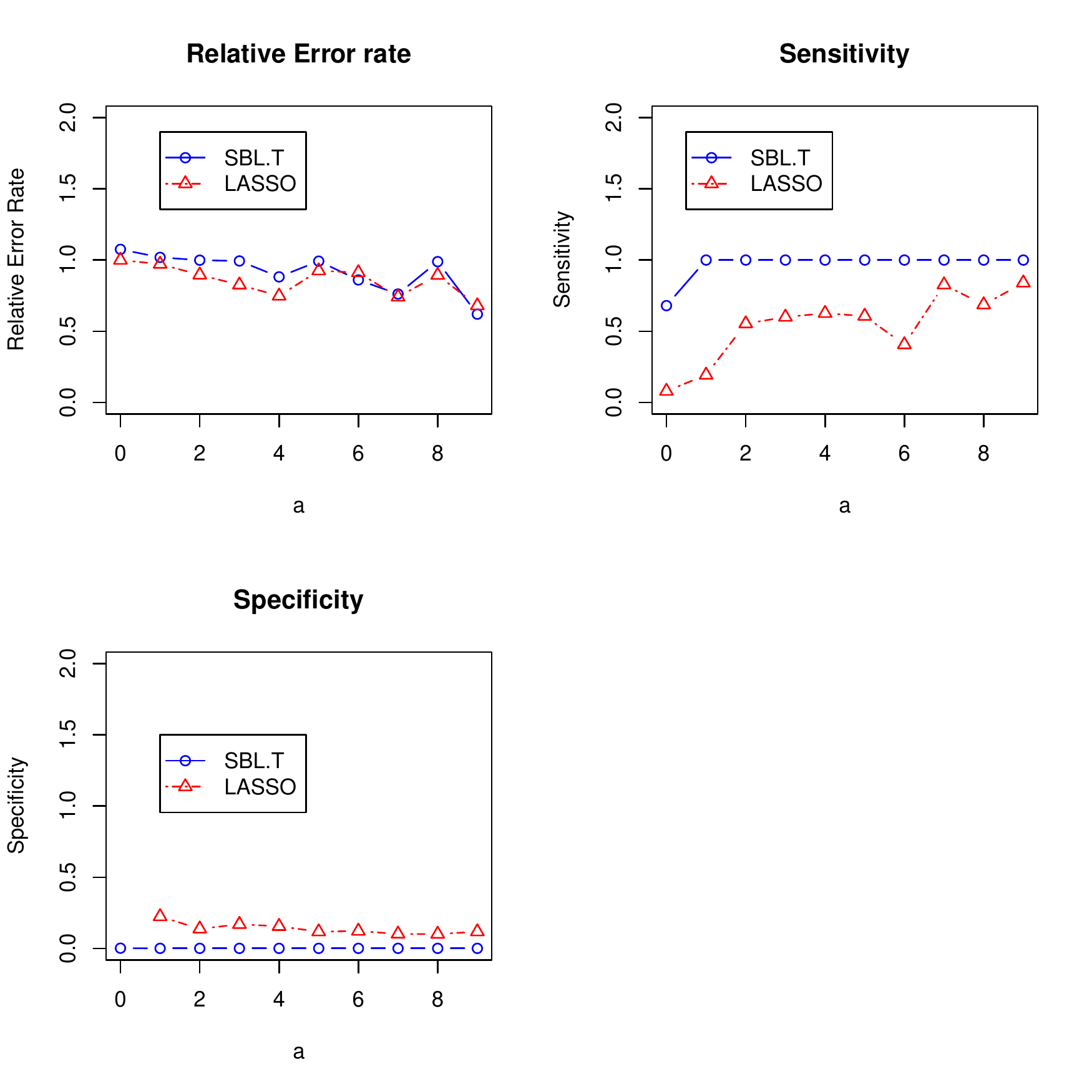}}\\[0pt]
\noindent\underline{Figure 9}: Sensitivity, specificity and relative error for SBL(thresholded) and lasso as function of $a$. $s=5$. The results of this simulation is generated using the riboflavin data}
\end{center}





\section{Proofs}\label{sec:proofs}
\subsection{Proof of Proposition \ref{prop1}}\label{proof:prop1}
\begin{proof}
For any $i\in\{1,\ldots,p\}$,
\[\ell(\gamma)\leq -\frac{1}{2}\log\det(C_\gamma)\leq -\frac{1}{2}\log\det\left(\sigma^2\Id_n +\gamma_{i}x_{i}x_{i}'\right)\downarrow-\infty,\]
as $\gamma_{i}\to\infty$. This together with the continuity of $\gamma\mapsto \ell(\gamma)$ imply the existence of a maximizer. For any such maximizer $\hat\gamma$, consider the vector $\gamma$ such that the $j$-th component of $\gamma$ is free to vary and the remaining components $\gamma_{-j}$ are fixed to  $\hat\gamma_{-j}$. Then we write $C_\gamma=C_j + \gamma_jx_jx_j'$ and use the matrix identity $(A+uu')^{-1}=A^{-1} -\frac{A^{-1}uu'A^{-1}}{1+u'A^{-1}u}$ to deduce that 
\[C_\gamma^{-1} =C_j^{-1} -\frac{\gamma_jC_j^{-1}x_jx_j'C_j^{-1}}{1+\gamma_j x_j'C_j^{-1}x_j}.\]
Therefore,
\[\ell(\gamma)=-\frac{1}{2}\log\det\left(C_j + \gamma_jx_jx_j'\right) +\frac{1}{2} \frac{\gamma_j\left(x_j'C_j^{-1}y\right)^2}{1+\gamma_j x_j'C_j^{-1}x_j}-\frac{1}{2}y'C_j^{-1}y.\]
Since $C_j$ does not depend on $\gamma_j$, we easily see that $\gamma_j\mapsto \ell(\gamma)$ is differentiable on $(0,\infty)$ and
\begin{eqnarray*}
\frac{\partial}{\partial \gamma_j}\ell(\gamma) &=&-\frac{1}{2} x_j'C_j^{-1}x_j +\frac{1}{2}\frac{\gamma_j\left(x_j'C_j^{-1}x_j\right)^2}{1+\gamma_jx_j'C_j^{-1}x_j} + \frac{1}{2}\frac{\left(x_j'C_j^{-1}y\right)^2}{\left(1+\gamma_jx_j'C_j^{-1}x_j\right)^2}\\
&=&\frac{1}{2}\frac{\left(x_j'C_j^{-1}y\right)^2 -\gamma_j\left(x_j'C_j^{-1}x_j\right)^2-x_j'C_j^{-1}x_j }{\left(1+\gamma_jx_j'C_j^{-1}x_j\right)^2}
\end{eqnarray*}
If $\left(x_j'C_j^{-1}y\right)^2\leq x_j'C_j^{-1}x_j$,  $\frac{\partial}{\partial \gamma_j}\ell(\gamma)<0$ and $\gamma_j\mapsto \ell(\gamma)$ attains its maximum at $0$. Similarly if $\left(x_j'C_j^{-1}y\right)^2> x_j'C_j^{-1}x_j$, it is easy to check that $\gamma_j\mapsto \ell(\gamma)$ attains its maximum at $(\left(x_j'C_j^{-1}y\right)^2-x_j'C_j^{-1}x_j)/\left(x_j'C_j^{-1}x_j\right)^2$.  Now if $\hat\gamma_j$ differs from the maximizer just found, we can improve on the likelihood by setting $\hat\gamma_{j}$ equal to that maximizer,  which would be a contradiction. Hence the result.

\end{proof}

\subsection{Proof of Proposition \ref{prop:prob}}\label{proof:prop:prob}
\begin{proof}
Recall that $\I=\{1\leq j\leq p:\;\beta_{\star,j}\neq 0\}$ is the sparsity structure of $\beta_\star$. For $\gamma\in\Theta$, and $1\leq j\leq p$, we define $\I_0\eqdef \I\cap \I_\gamma\setminus\{j\}$, and $\I_1\eqdef \I^c\cap \I_\gamma\setminus\{j\}$, where in order to keep  the notation easy, we omit the dependence of $\I_0$ and $\I_1$ on $(\gamma,j)$. We will also write $X_{\I_0}$ (resp. $X_{\I_1}$) to denote the matrix obtained by collecting the columns of $X$ whose indexes belong to $\I_0$ (resp. $\I_1$).  We define 
\begin{eqnarray*}C_{j,\gamma}&\eqdef &\sigma^2\Id_n +\sum_{k\in \I_\gamma\setminus\{j\}} \gamma_kx_kx_k'\\
&=&\sigma^2\Id_n +\sum_{k\in\I_0}\gamma_kx_kx_k' +\sum_{k\in\I_1}\gamma_kx_kx_k'.\end{eqnarray*}
By the Woodbury matrix identity and the assumption $X_{\I_0}'X_{\I_1} =0$, we get:
\[C_{j,\gamma}^{-1}=\frac{1}{\sigma^2}\Id_n -\frac{1}{\sigma^4}X_{\I_0}\left(\Gamma^{-1}_{\I_0}+\frac{1}{\sigma^2}X_{\I_0}'X_{\I_0}\right)^{-1}X_{\I_0}' -\frac{1}{\sigma^4}X_{\I_1}\left(\Gamma^{-1}_{\I_1}+\frac{1}{\sigma^2}X_{\I_1}'X_{\I_1}\right)^{-1} X_{\I_1}'.\]
Hence, for $k\in\I$, and using the fact that $j\notin \I$, we have
\[x_j'C_{j,\gamma}^{-1}x_k=0,\;\mbox{ and }\; x_j'C_{j,\gamma}^{-1}x_j = \frac{1}{\sigma^2}\|x_j\|^2.\]
Therefore, if $Y=X\beta_\star +\epsilon$, we get
\[x_j'C_{j,\gamma}^{-1}Y = \frac{1}{\sigma^2}\seq{x_j,\epsilon}\sim\textbf{N}\left(0,\sigma^2\|x_j\|^2\right).\]
Now, the matrix $C_j$ defined in Proposition \ref{prop1} is $C_j=C_{j,\hat\gamma_n}$. Hence
\[\PP\left[\hat\gamma_{n,j}=0\right] =\PP\left[\left(x_j'C_{j,\hat\gamma_n}^{-1}Y\right)^2\leq x_j'C_{j,\hat\gamma_n}^{-1}x_j\right]=  \PP\left[Z^2\leq 1\right],\]
where $Z\sim\textbf{N}(0,1)$. Hence the result.
\end{proof}

\subsection{Proof of Theorem \ref{thm1}}\label{proof:thm1}
\begin{proof}
Under H\ref{A2}, $x_jC_{j,\hat\gamma_n}^{-1}Y=\seq{x_j,Y}/\sigma^2$, and $x_jC_{j,\hat\gamma_n}^{-1}x_j= \|x_j\|^2/\sigma^2$. Hence
\[\tilde\gamma_j=\left\{ \begin{array}{ll} \frac{\seq{x_j,Y}^2-\sigma^2\|x_j\|^2}{\|x_j\|^2} & \mbox{ if }\;\; \seq{x_j,Y}^2>\sigma^2\|x_j\|^2(1+z_\star)\\
0 &\mbox{ otherwise}.\end{array}\right.\]
Similarly, under H\ref{A1} $\hat\beta_{n,j}$ has the explicit form $\hat\beta_{n,j} =\frac{\seq{Y,X_j}}{\|x_j\|^2+\frac{\sigma^2}{\hat\gamma_{n,j}}}$. It follows that
\[\tilde\beta_{n,j}=\left\{\begin{array}{ll} \frac{\seq{Y,x_j}}{\|x_j\|^2+\frac{\sigma^2}{\hat\gamma_{n,j}}} & \mbox{ if } \; \seq{Y,x_j}^2>\sigma^2\|x_j\|^2(1+z_\star)\\
0 & \mbox{ otherwise } .\end{array}\right.\]
Again using  the orthogonality assumption of $X$, we obtain $\seq{Y,x_j}=\beta_{\star,j}\|x_j\|^2 + \seq{\epsilon,X_j}$. We set $t_j\eqdef \seq{\epsilon,X_j}$. Then it follows that
\[\tilde\beta_{n,j}-\beta_{\star,j}=\left\{\begin{array}{ll} 0 & \mbox{ if } \beta_{\star,j}=0,\; \mbox{ and } \left(\frac{t_j}{\|x_j\|^2}\right)^2\leq \frac{\sigma^2}{\|x_j\|^2}(1+z_\star)\\
\frac{t_j}{\|x_j\|^2+\frac{\sigma^2}{\tilde \gamma_{n,j}}} & \mbox{ if } \beta_{\star,j}=0, \mbox{ and } \left(\frac{t_j}{\|x_j\|^2}\right)^2> \frac{\sigma^2}{\|x_j\|^2}(1+z_\star)\\
-\beta_{\star,j} & \mbox{ if } \beta_{\star,j}\neq 0 \mbox{ and } \left(\frac{t_j}{\|x_j\|^2}+\beta_{\star,j}\right)^2\leq \frac{\sigma^2}{\|x_j\|^2}(1+z_\star)\\
\frac{\|x_j\|^2\beta_{\star,j} +t_j}{\|x_j\|^2 + \frac{\sigma^2}{\tilde \gamma_{n,j}}}-\beta_{\star,j} & \mbox{ if } \beta_{\star,j}\neq 0 \mbox{ and } \left(\frac{t_j}{\|x_j\|^2}+\beta_{\star,j}\right)^2> \frac{\sigma^2}{\|x_j\|^2}(1+z_\star).\end{array}\right.\]

Suppose that $j\in \I=\I_{\gamma_\star}$ and $\left(\frac{t_j}{\|x_j\|^2}+\beta_{\star,j}\right)^2\leq \frac{\sigma^2}{\|x_j\|^2}(1+z_\star)$. Then with $Z_j\eqdef \frac{t_j}{\sigma\|x_j\|}$,
\[|\beta_{\star,j}|\leq \left|\frac{t_j}{\|x_j\|^2}\right| + \frac{\sigma}{\|x_j\|}\sqrt{1+z_\star}=\frac{\sigma}{\|x_j\|}\left(|Z_j| + \sqrt{1+z_\star}\right).\]
Hence for such index $j$, 
\begin{equation}\label{theo1eq1}
\beta_{\star,j}^2\leq \frac{2\sigma^2}{\|x_j\|^2}\left(Z_j^2 + 1+z_\star\right).\end{equation}

But for $j\in\I_{\gamma_\star}$, such that $\left(\frac{t_j}{\|x_j\|^2}+\beta_{\star,j}\right)^2> \frac{\sigma^2}{\|x_j\|^2}(1+z_\star)$, $\tilde\gamma_{n,j}=\left(\frac{t_j}{\|x_j\|^2}+\beta_{\star,j}\right)^2- \frac{\sigma^2}{\|x_j\|^2}$. Using this with some easy algebra, we obtain that for such index $j$,
\begin{equation}\label{theo1eq2}
\frac{\|x_j\|^2\beta_{\star,j} +t_j}{\|x_j\|^2 + \frac{\sigma^2}{\tilde \gamma_{n,j}}}-\beta_{\star,j}=\frac{t_j}{\|x_j\|^2}-\frac{\sigma^2}{t_j+\|x_j\|^2\beta_{\star,j}}.\end{equation}
It follows that
\begin{equation}\label{theo1eq3}
\left|\frac{\|x_j\|^2\beta_{\star,j} +t_j}{\|x_j\|^2 + \frac{\sigma^2}{\tilde \gamma_{n,j}}}-\beta_{\star,j}\right|\leq \frac{\sigma}{\|x_j\|}\left(1+ |Z_j|\right).\end{equation}
With (\ref{theo1eq1}) and (\ref{theo1eq3}), we get
\[\sum_{j\in\I_{\gamma_\star}} \left(\tilde \beta_{n,j}-\beta_{\star,j}\right)^2\leq \frac{2\sigma^2}{cn}\sum_{j\in\I_{\gamma_\star}} (1+z_\star + Z_j^2),\]
where $c=\min_{1\leq i\leq p}\|x_j\|^2/n$. Set $s\eqdef |I_{\gamma_\star}|$. By \cite{teicher84}~Lemma 5, $\PE(|Z_j^2-1|^k)\leq k!2^{k-2}$, $k>2$. Hence by Bernstein's inequality (see e.g. \cite{vandervaartetwellner96}~Lemma 2.2.11), we conclude that
\begin{multline*}
\PP\left[\sum_{j\in\I_{\gamma_\star}} (1+z_\star + Z_j^2)>2(1+z_\star)s\right]\leq \PP\left[\sum_{j\in\I_{\gamma_\star}} (Z_j^2-1)>z_\star s\right]\\
\leq \exp\left(-\frac{z_\star^2s^2}{4(1+z_\star s)}\right)\leq \exp\left(-\frac{z_\star s}{8}\right)\leq \frac{1}{p^{c_0s/8}}.\end{multline*}
Hence with probability at least $1-\frac{1}{p^{c_0s/8}}$,
\begin{equation}\label{theo1eqfinal1}
\sum_{j\in \I_{\gamma_\star}} \left(\tilde\beta_{n,j}-\beta_{\star,j}\right)^2\leq \frac{4\sigma^2}{cn}\left(1+z_\star\right)s\leq \frac{4(1+c_0)}{c}\frac{\sigma^2 s\log p}{n}.\end{equation}

 On the other hand, from (\ref{theo1eq2}), $\frac{t_j}{\|x_j\|^2+\frac{\sigma^2}{\tilde \gamma_{n,j}}}=\frac{t_j}{\|x_j\|^2}-\frac{\sigma^2}{t_j}$, hence
 \begin{multline*}
 \sum_{j\notin \I_{\gamma_\star}} \left(\tilde\beta_{n,j}-\beta_{\star,j}\right)^2=\sum_{j\notin\I_{\gamma_\star},\,Z_j^2>1+z_\star} \frac{\sigma^2}{\|x_j\|^2}\left(Z_j-\frac{1}{Z_j}\right)^2\leq \frac{\sigma^2}{cn}\sum_{j\notin\I_{\gamma_\star}} Z_j^2\textbf{1}_{\{Z_j^2>1+z_\star\}}\\
 \leq\frac{\sigma^2}{cn}\sum_{j=1}^p Z_j^2\textbf{1}_{\{Z_j^2>1+z_\star\}}.
 \end{multline*}
Now for any $\kappa\in (0,1/2)$, $a>0$, and by Markov's inequality
\begin{eqnarray}\label{theo1eq4}
\PP\left[\sum_{j=1}^p Z_j^2\textbf{1}_{\{Z_j^2>1+z_\star\}}>a\right]&=& \PP\left[\exp\left(\sum_{j=1}^p \kappa Z_j^2\textbf{1}_{\{Z_j^2>1+z_\star\}}\right) >e^{a \kappa }\right]\nonumber\\
&\leq &\exp\left[-a \kappa  +  p\log\PE\left[\exp\left(\kappa Z_1^2\textbf{1}_{\{Z_1^2>1+z_\star\}}\right)\right]\right].\end{eqnarray}
We calculate that
\begin{eqnarray*}
\PE\left[\exp\left(\kappa Z_1^2\textbf{1}_{\{Z_1^2>1+z_\star\}}\right)\right]&=&2\int_0^{\sqrt{1+z_\star}} \frac{e^{-x^2/2}}{\sqrt{2\pi}}\rmd x +2 \int_{\sqrt{1+z_\star}}^\infty \frac{e^{-\frac{1}{2}(1-2\kappa)x^2}}{\sqrt{2\pi}}\rmd x\\
&\leq & 1 + 2 \int_{\sqrt{1+z_\star}}^\infty \frac{e^{-\frac{1}{2}(1-2\kappa)x^2}}{\sqrt{2\pi}}\rmd x\\
&\leq & 1 + \frac{\exp\left(-\frac{z_\star(1-2\kappa)}{2}\right) }{1-2\kappa},\end{eqnarray*}
where the last inequality uses some easy algebra and the well known bound on the Gaussian cdf: $\int_t^\infty \frac{1}{\sqrt{2\pi}}e^{-x^2/2a^2}\rmd x\leq \frac{a^2 e^{-t^2/2a^2}}{t\sqrt{2\pi}}$, valid for all $t>0$. With $z_\star=c_0\log p$, We deduce that
\[p\log \PE\left[\exp\left(\kappa Z_1^2\textbf{1}_{\{Z_1^2>1+z_\star\}}\right)\right]\leq p\log\left(1+  \frac{p^{-\frac{c_0(1-2\kappa)}{2}}}{(1-2\kappa)}\right)\leq\frac{p^{c_0\kappa}}{1-2\kappa}.\]
Hence with $a=\frac{2-2\kappa}{\kappa}\frac{p^{c_0\kappa}}{1-2\kappa}\leq 4\kappa^{-1}p^{c_0\kappa}$, (\ref{theo1eq4}) gives
\[\PP\left[\sum_{j=1}^p Z_j^2\textbf{1}_{\{Z_j^2>1+z_\star\}}>a\right]\leq \exp\left(-a\kappa + \frac{p^{c_0\kappa}}{1-2\kappa}\right) \leq \exp\left(-p^{c_0\kappa}\right).\]
We conclude that with probability at least $1-\exp(-p^{c_0\kappa})$, $\sum_{j=1}^p Z_j^2\textbf{1}_{\{Z_j^2>1+z_\star\}}\leq a\leq 4\kappa^{-1}p^{c_0\kappa}$, so that
\begin{equation}\label{theo1eqfinal2}
 \sum_{j\notin \I_{\gamma_\star}} \left(\tilde\beta_{n,j}-\beta_{\star,j}\right)^2\leq \frac{4\sigma^2}{cn}\kappa^{-1}p^{c_0\kappa},\end{equation}
 with probability at least $1-\exp(-p^{c_0\kappa})$. Combining (\ref{theo1eqfinal1}) and (\ref{theo1eqfinal2}) it follows that
 \[\|\tilde\beta_n-\beta_\star\|^2_2\leq \frac{4\sigma^2}{c n}\left((1+c_0)s\log p+ \frac{p^{c_0\kappa}}{\kappa}\right),\]
 with probability at least $1-\frac{1}{p^{c_0s/8}}-\exp(-p^{c_0\kappa})$. Finally since $\log s>1$, we can take $\kappa=\log (s)/(c_0\log(p))\in (0,1/2)$ to achieve $s=p^{c_0\kappa}$. With this choice, 
 \[\frac{p^{c_0\kappa}}{\kappa}=\frac{s\log(p)}{c_0\log(s)}\leq \frac{s\log(p)}{c_0},\]
and the theorem follows easily.

\end{proof}

\subsection{Proof of Corollary \ref{coro1}}\label{proof:coro1}
\begin{proof}
Recall that $\I=\I_{\gamma_\star}=\{1\leq j\leq p:\; \beta_{\star,j}\neq 0\}$. We write $u_{\I_{\gamma_\star}}=(u_j,\;j\in\I_{\gamma_\star})$, and $u_{\I^c_{\gamma_\star}}=(u_j,\;j\notin\I_{\gamma_\star})$. It is clear that whenever (\ref{cond:signal}) holds and $|\tilde\beta_{n,j}-\beta_{\star,j}|\leq \sqrt{M\sigma^2 s\log(p)/n}$, we have $\textsf{sign}(\tilde\beta_{n,j})=\textsf{sign}(\beta_{\star,j})$. Since $|\tilde\beta_{n,j}-\beta_{\star,j}|\leq \|\tilde\beta_n-\beta_\star\|_2$,  we conclude that $\textsf{sign}(\tilde\beta_{n,\I_{\gamma_\star}})=\textsf{sign}(\beta_{\star,\I_{\gamma_\star}})$, with probability at least  $1-\frac{1}{p^{(c_0 s)/8}}-\frac{1}{\exp(s)}$.

For the other part, it follows from the definition of $\tilde\beta_n$ that for $\beta_{\star,j}=0$, $\textsf{sign}(\tilde\beta_{n,j})\neq 0$ implies that $Z_j^2\geq 1+z_\star$. Hence
\begin{eqnarray*}
\PP\left(\textsf{sign}(\tilde\beta_{n,\I^c_{\gamma_\star}})\neq \textsf{sign}(\beta_{\star,\I^c_{\gamma_\star}})\right)&\leq &\sum_{j=1}^p \PP\left(Z_j^2>1+z_\star\right)= \sum_{j=1}^p 2\PP\left(Z_j>\sqrt{1+z_\star}\right)\\
&\leq & \sum_{j=1}^p e^{-\frac{1}{2}(1+z_\star)}\leq \exp\left(\log p-\frac{c_0}{2}\log p\right)\\
&\leq &\frac{1}{p^{\frac{c_0}{2}-1}}.\end{eqnarray*}
The results follows.

\end{proof}

\subsection{Proof of Proposition \ref{prop2}}\label{proof:prop2}
\begin{proof}
We fix an arbitrary $j\in\{1,\ldots,p\}$. We define $x_k=\|x_j\|^2\gamma_{j}^{(k)}$, where we omit the dependence on $j$ to keep the notation simple. It follows from (\ref{recEM}) that
\[x_{k+1}=B \left(\frac{x_k}{\sigma^2+x_k}\right)^2 + \frac{\sigma^2 x_k}{\sigma^2+x_k}=\Psi(x_k),\]
where $B=\seq{y,X_j}^2/\|x_j\|^2$, and 
\[\Psi(x)=B \left(\frac{x}{\sigma^2+x}\right)^2 + \frac{\sigma^2 x}{\sigma^2+x}.\]
Notice that for all $x\geq 0$, $\Psi(x)\in [0,\sigma^2+B]$. Hence the sequence $\{x_k,\,k\geq 0\}$ is bounded. The equation $\Psi(x)=x$ is equivalent to $x^2(x+\sigma^2)=x^2B$. If $B\leq \sigma^2$, $\Psi(x)=x$ has a unique solution $x=0$. If $B> \sigma^2$, then $\Psi(x)=x$ has two solutions $x=0$ and $x=B-\sigma^2$. The derivatives of $\Psi$ are given by
\[\Psi'(x)=\frac{\sigma^4}{(\sigma^2+x)^2} + \frac{2x\sigma^2 B}{(\sigma^2+x)^3},\;\;\Psi^{''}(x)=\frac{-2\sigma^2x(\sigma^2+2B)+2\sigma^4(B-\sigma^2)}{(\sigma^2+x)^4}.\]
We consider two cases
\begin{enumerate}
\item Case 1: $B\leq \sigma^2$.\;\;\; Then $\Psi^{''}(x)\leq 0$ for all $x\geq 0$. Hence $\Psi$ is concave. This implies that for all $x\geq 0$,
\[\Psi(x)\leq \Psi(0) +\Psi'(0)x=x.\]
This implies that $x_k=\Psi(x_{k-1})\leq x_{k-1}$. This means that the sequence $\{x_k,\;k\geq 0\}$ is bounded and non-increasing, hence has a limit $x_\star$. By continuity of $\Psi$, the limit point $x_\star$ satisfies $\Psi(x_\star)=x_\star$. Hence $x_\star=0$, since we have seen above that $0$ is the only fixed-point of $\Psi$ when $B\leq \sigma^2$.
\item Case 2: $B> \sigma^2$:\;\;\;\; Then $\Psi^{''}(0)>0$, and by Taylor expansion, in a neighborhood of $0$, we have $\Psi(x)\geq \Psi(0) + \Psi'(0)x=x$ for all $x>0$ small enough. If $x_\star=B_j-\sigma^2$ denotes the unique positive fixed point of $\Psi$, we can conclude that for all $x\in [0,x_\star]$, $\Psi(x)\geq x$, and for $x>x_\star$, $\Psi(x)< x$. Therefore, if $x_0\in [0,x_\star]$, then $\{x_k,\, k\geq 0\}$ is increasing and bounded, hence converges to the unique positive fixed point $x_\star$ (recall that $x_0>0$). Whereas, if $x_0>x_\star$, then $\{x_k,\, k\geq 0\}$ is decreasing and bounded, hence converges to the unique positive fixed point $x_\star$.
\end{enumerate}

\end{proof}

\bibliographystyle{ims}
\bibliography{biblio}

\begin{thebibliography}{13}
\expandafter\ifx\csname natexlab\endcsname\relax\def\natexlab#1{#1}\fi
\expandafter\ifx\csname url\endcsname\relax
  \def\url#1{\texttt{#1}}\fi
\expandafter\ifx\csname urlprefix\endcsname\relax\def\urlprefix{URL }\fi

\bibitem[{Balakrishnan and Madigan(2010)}]{balaetmadigan09}
\textsc{Balakrishnan, S.} and \textsc{Madigan, D.} (2010).
\newblock Priors on the variance in sparse bayesian learning: the demi-bayesian
  lasso.
\newblock In \textit{Frontiers of Statistical Decision Making and Bayesian
  Analysis: In Honor of James O. Berger} (M.-H. Chen, P.~Muller, D.~Sun and
  K.~Ye, eds.). Springer, New York.

\bibitem[{Bickel et~al.(2009)Bickel, Ritov and Tsybakov}]{bickeletal09}
\textsc{Bickel, P.~J.}, \textsc{Ritov, Y.} and \textsc{Tsybakov, A.~B.} (2009).
\newblock Simultaneous analysis of lasso and {D}antzig selector.
\newblock \textit{Ann. Statist.} \textbf{37} 1705--1732.

\bibitem[{B{\"u}hlmann and van~de Geer(2011)}]{buhlmannetdegeer11}
\textsc{B{\"u}hlmann, P.} and \textsc{van~de Geer, S.} (2011).
\newblock \textit{Statistics for high-dimensional data}.
\newblock Springer Series in Statistics, Springer, Heidelberg.
\newblock Methods, theory and applications.

\bibitem[{Faul and Tipping(2002)}]{faulettipping02}
\textsc{Faul, A.~C.} and \textsc{Tipping, M.~E.} (2002).
\newblock Analysis of sparse bayesian learning.
\newblock \textit{Advances in Neural Information Processing Systems}
  \textbf{14} 383--389.

\bibitem[{Friedman et~al.(2010)Friedman, Hastie and Tibshirani}]{glmnet}
\textsc{Friedman, J.~H.}, \textsc{Hastie, T.} and \textsc{Tibshirani, R.}
  (2010).
\newblock Regularization paths for generalized linear models via coordinate
  descent.
\newblock \textit{Journal of Statistical Software} \textbf{33} 1--22.
\newline\urlprefix\url{http://www.jstatsoft.org/v33/i01}

\bibitem[{Golub and Van~Loan(2013)}]{golubetvanloan13}
\textsc{Golub, G.} and \textsc{Van~Loan, C.~F.} (2013).
\newblock \textit{Matrix Computations, 4th Ed.}
\newblock John Hopkins University Press, Baltimore.

\bibitem[{Meinhausen and Yu(2009)}]{meinhausenetyu09}
\textsc{Meinhausen, N.} and \textsc{Yu, B.} (2009).
\newblock Lasso-type recovery of sparse representations for high-dimensional
  data.
\newblock \textit{Ann. Statist.} \textbf{37} 246--270.

\bibitem[{O'Hara and Sillanp{\"a}{\"a}(2009)}]{oharaetsillanpaa09}
\textsc{O'Hara, R.~B.} and \textsc{Sillanp{\"a}{\"a}, M.~J.} (2009).
\newblock A review of {B}ayesian variable selection methods: what, how and
  which.
\newblock \textit{Bayesian Anal.} \textbf{4} 85--117.

\bibitem[{Teicher(1984)}]{teicher84}
\textsc{Teicher, H.} (1984).
\newblock Exponential bounds for large deviations of sums of unbounded random
  variables.
\newblock \textit{Sankhy\=a Ser. A} \textbf{46} 41--53.

\bibitem[{Tibshirani(1996)}]{tib96}
\textsc{Tibshirani, R.} (1996).
\newblock Regression shrinkage and selection via the lasso.
\newblock \textit{Journal of the Royal Statistical Society. Series B}
  \textbf{58} 267--288.

\bibitem[{Tipping(2001)}]{tipping01}
\textsc{Tipping, M.} (2001).
\newblock Sparse {B}ayesian learning and the relevance vector machine.
\newblock \textit{JMLR} \textbf{1} 211--244.

\bibitem[{van~der Vaart and Wellner(1996)}]{vandervaartetwellner96}
\textsc{van~der Vaart, A.~W.} and \textsc{Wellner, J.~A.} (1996).
\newblock \textit{Weak convergence and empirical processes}.
\newblock Springer Series in Statistics, Springer-Verlag, New York.
\newblock With applications to statistics.

\bibitem[{Wipf and Rao(2004)}]{wipfetrao04}
\textsc{Wipf, D.~P.} and \textsc{Rao, B.~D.} (2004).
\newblock Sparse {B}ayesian learning for basis selection.
\newblock \textit{IEEE Trans. Signal Process.} \textbf{52} 2153--2164.

\end{thebibliography}

\end{document}